\documentclass[a4paper,onecolumn,10pt]{quantumarticle}
\pdfoutput=1
\usepackage{graphicx} 
\usepackage{amssymb}
\usepackage[utf8]{inputenc}
\usepackage[T1]{fontenc}
\usepackage{amsfonts}
\usepackage{hyperref}
\usepackage{epsfig,pstricks,graphicx}
\usepackage{bm}
\usepackage{physics}
\usepackage{amsthm}
\usepackage{mathtools}
\usepackage[numbers]{natbib}
\def\id{{\rm 1\kern-.22em l}}

\newcommand{\ii}{\operatorname{in}}
\newcommand{\ou}{\operatorname{out}}
\newcommand{\rme}{\rho_{\operatorname{ME}}}
\newcommand{\ud}{\operatorname{U}(d)}
\newtheorem{lem}{Lemma}

\newtheorem{theorem}{Theorem}
\newtheorem{proposition}{Proposition}

\begin{document}
\author{Marcin Markiewicz}
\affiliation{Institute of Theoretical and Applied Informatics, Polish Academy of
Sciences, ul. Ba{\l}tycka 5, 44-100 Gliwice, Poland} \affiliation{International
Centre for Theory of Quantum Technologies (ICTQT), University of Gdansk, 80-308
Gdansk, Poland}
\orcid{0000-0002-8983-9077}
\email{mmarkiewicz@iitis.pl}
\author{{\L}ukasz Pawela} \affiliation{Institute of Theoretical and Applied
Informatics, Polish Academy of Sciences, ul. Ba{\l}tycka 5, 44-100 Gliwice,
Poland}
\orcid{0000-0002-0476-7132}
\author{Zbigniew Pucha{\l}a} \affiliation{Institute of Theoretical and Applied
Informatics, Polish Academy of Sciences, ul. Ba{\l}tycka 5, 44-100 Gliwice,
Poland}
\orcid{0000-0002-4739-0400}
\title{Choi-level twirling of quantum channels: finite constructions and non-compact transformations}

\begin{abstract}
Twirling, i.e.\ averaging over symmetry actions, is a standard tool for reducing quantum states and channels to a symmetry-invariant form. 
We study channel twirling from the perspective of the channel--state duality and provide a constructive Choi-level description of the averaging map induced by arbitrary input/output representations.
Our main technical result concerns the collective setting: for $\pi^{\ii}(U)=U^{\otimes t_{\ii}}$ and $\pi^{\ou}(U)=U^{\otimes t_{\ou}}$, we introduce a partial-transpose reduction that removes the contragredient action and converts the mixed (walled Brauer) channel twirl into an \emph{ordinary} Schur--Weyl twirl of the partially transposed Choi operator under $U^{\otimes(t_{\ii}+t_{\ou})}$, enabling explicit permutation-based formulas without constructing walled Brauer idempotents or mixed Schur transforms.
Beyond compact symmetries, we extend channel twirling to reductive, generally non-unitary groups via Cartan decomposition and obtain an invariant-sector decomposition of the averaged Choi operator with weights determined solely by the Abelian Cartan component.
Finally, we provide two finite realizations of channel averaging: a ``dual'' implementation as a convex mixture of unitary-$1$-design channels acting on invariant sectors, and a design-like reconstruction showing that weighted group $t$-designs induce channel $t$-designs for $t=t_{\ii}+t_{\ou}$.
\end{abstract}

\maketitle

\section{Introduction}

\subsection{General Motivation}
Averaging procedures based on symmetry groups play a central role in quantum information theory \cite{Bartlett03a, Bartlett07}. 
For quantum states, uniform averaging (twirling) with respect to a group action provides a systematic way of reducing descriptions to symmetry-invariant data, with applications ranging from entanglement theory to randomized benchmarking and noise-invariant encoding \cite{Zanardi97, Eggeling01, Bae19,Gross21, Nakata21, Markiewicz23}.
For quantum channels, averaging is typically defined operationally via uniform pre- and post-processing by symmetry transformations, leading to effective noise models such as depolarizing or  stochastic channels \cite{Bartlett07,Miatto12,Winter21,Graydon22, Kong22, Nechita25}.

Despite its widespread use, channel averaging is most often treated at the level of operational definitions and specific examples.
The Choi--Jamio\l{}kowski isomorphism provides a natural representation of quantum channels as bipartite operators. While this framework is invaluable for tasks such as defining operational distance measures~\cite{puchala2011experimentally} and discriminating quantum operations~\cite{krawiec2020discrimination}, a general constructive Choi-level formula that works uniformly for arbitrary collective tensor actions and input/output transformations remains largely absent from the literature. Consequently, practical evaluation often relies on case-by-case representation-theoretic machinery.
Existing approaches either focus on special classes of symmetries, such as finite or compact groups with irreducible representations, or rely on case-by-case constructions tailored to specific physical settings.

A natural expectation is that channel twirling should correspond to a projection onto symmetry-invariant subspaces of the Choi operator.
However, making this correspondence precise for arbitrary input and output transformations, and doing so in a way that is both conceptually transparent and practically usable, turns out to be nontrivial.
This difficulty becomes particularly apparent in collective scenarios, where channels act on multiple identical subsystems and the relevant symmetry algebra is no longer given by the symmetric group alone (as in the case of twirling of quantum states) \cite{Markiewicz23}, but by the commutant of the tensor product action $U^{\otimes p}\otimes \bar{U}^{\otimes q}$, namely the walled Brauer algebra \cite{Dipper08}.

\subsection{State of the art: covariant and equivariant channels, and mixed Schur--Weyl tools.}
Symmetry constraints on quantum channels are most commonly formulated in terms of covariance\footnote{It is important to distinguish two related but different notions:
\emph{covariance/equivariance} is a symmetry \emph{constraint} defining a class of channels,
whereas \emph{twirling} is an \emph{averaging procedure} that maps an arbitrary channel to a symmetric one.}.
A channel $\Phi:\mathcal{B}(\mathcal{H}_{\mathrm{in}})\to\mathcal{B}(\mathcal{H}_{\mathrm{out}})$ is called $\mathcal G$-covariant if there exist (typically unitary) representations $\pi_{\mathcal{G}}^{\textrm{in}}$ and $\pi_{\mathcal{G}}^{\textrm{out}}$ of the symmetry group $\mathcal G$, such that
\begin{equation}
\Phi\bigl(\pi_{\mathcal{G}}^{\textrm{in}}(g)\,\rho\,\pi_{\mathcal{G}}^{\textrm{in}}(g)^\dagger\bigr)
= \pi_{\mathcal{G}}^{\textrm{out}}(g)\,\Phi(\rho)\,\pi_{\mathcal{G}}^{\textrm{out}}(g)^\dagger
\qquad \forall g\in \mathcal G.
\end{equation}
It is a well-known consequence of the Choi--Jamio\l{}kowski
isomorphism that this covariance condition is equivalent to a commutation
relation at the Choi level \cite{Studzinski17,Winter21,Grinko2024}:
\begin{equation}
\bigl[\,
\pi_{\mathcal G}^{\textrm{out}}(g)\otimes \bar\pi_{\mathcal G}^{\textrm{in}}(g),
\, J_{\Phi}
\,\bigr]=0
\qquad \forall g\in\mathcal G,
\end{equation}
that is, the Choi operator must lie in the commutant of the induced
representation $\pi_{\mathcal G}^{\textrm{out}}\otimes
\bar\pi_{\mathcal G}^{\textrm{in}}$.

When the input and output representations are identical, or related in a direct
``same-action'' manner (e.g.\ $U^{\otimes p}$ on the input and $U^{\otimes q}$
on the output for $U\in \mathrm{U}(d)$), the literature often uses the term
\emph{equivariance} to emphasize that applying the symmetry to the input
corresponds to applying (essentially) the same symmetry to the output.
This distinction is especially natural in the collective setting of many
identical subsystems, where one studies channels whose Choi operators commute
with the tensor action $U^{\otimes q}\otimes \bar{U}^{\otimes p}$.

A substantial body of work analyzes the structure and applications of such symmetry-restricted channels.
For finite groups, a detailed characterization is available in the irreducibly covariant case, where representation-theoretic data yields explicit decompositions of the Choi operator and necessary and sufficient conditions for complete positivity and trace preservation \cite{Studzinski17, Mozrzymas18}.
For continuous unitary symmetries and multipartite settings, mixed Schur--Weyl duality identifies the commutant of $U^{\otimes q}\otimes \bar{U}^{\otimes p}$ with the walled Brauer algebra, providing a powerful language for block diagonalization and parametrization \cite{Nguyen23, Grinko2024,Studzinski25}.
This perspective underlies modern optimization frameworks: unitary-equivariant semidefinite programs can be reduced to linear programs by expressing feasible operators in a walled Brauer basis and exploiting additional symmetry constraints \cite{Grinko23, Grinko2024}.
On the quantum-algorithmic side, efficient circuit constructions for the mixed Schur transform enable implementations and manipulations of unitary-equivariant channels in a variety of operational settings \cite{Nguyen23,Nguyen24}.

In parallel, channel twirling is a standard operational tool, defined by averaging pre- and post-processing operations over a group action \cite{Bartlett07,Miatto12,Winter21,Graydon22, Kong22, Nechita25}.
Connections to unitary designs are well understood in several important cases.
In particular, twirling by a unitary $1$-design yields a stochastic channel, although the resulting channel may depend on the specific choice of the design \cite{Graydon22}.
Much of the literature, however, focuses on concrete symmetry models or assumes the invariant decomposition a priori, rather than providing a general and constructive description of channel twirling directly at the level of the Choi operator.

Despite extensive work on covariant channels and mixed Schur--Weyl duality \cite{Nguyen23, Grinko23, Grinko23gt,Grinko2024,Studzinski25,Nechita25}, a general and constructive description of channel twirling at the Choi level, especially beyond compact symmetries, has been missing. Moreover in much of the literature on this topic the symmetry is imposed as a structural constraint on $\Phi$ (equivariant/covariant channels),
whereas here we focus on the \emph{averaging map} that sends an arbitrary $\Phi$ to its symmetry-averaged counterpart and we make this map constructive at the Choi level.

\subsection{Outline}

In this work, we develop a unified framework for averaging quantum channels
based on the channel--state duality. Starting from the standard definition of
channel twirling as uniform pre- and post-processing by symmetry operations,
we show that channel averaging can be expressed exactly as a group twirl acting
directly on the Choi operator. More precisely, the Choi operator of a twirled
channel is obtained by averaging the original Choi operator with respect to the
induced representation on the joint output--input space.

Although this correspondence between operational twirling and a Choi-level
group action appears implicitly or explicitly in previous works, we elevate it
to a central structural principle. In our approach, the Choi-level formulation
is not merely a reformulation, but the fundamental tool that enables a uniform
treatment of arbitrary representations of  input and output symmetry transformations, collective
tensor actions, and non-compact operations.

Viewed in this way, channel twirling becomes an explicit projection onto the
commutant of the induced symmetry representation, placing channel averaging
exactly within the framework of generalized Schur--Weyl duality \cite{Brundan08, Marvian14, Zhang15, Gross21}. This
structural perspective is what allows us to identify invariant-sector decompositions, derive finite-sum formulas, and extend the formalism beyond the
compact unitary setting.

In the important collective setting, where the symmetry group acts identically on multiple input and output subsystems, the commutant algebra is given by the walled Brauer algebra.
To avoid the technical complexity associated with constructing explicit Brauer idempotents, we introduce a partial-transpose reduction that maps channel twirling to an ordinary Schur--Weyl twirl of a partially transposed Choi operator.
As a result, the evaluation of collective channel twirls can be reduced to standard permutation-based formulas,
enabling explicit formulas in terms of permutation operators and Schur projectors, without constructing walled Brauer idempotents or mixed Schur transforms.

Genuinely beyond the symmetries represented by compact groups, we extend our framework to averaging over arbitrary finite-dimensional representations of reductive, generally non-compact groups.
Using Cartan decomposition, we show that the resulting channel twirl admits a decomposition into invariant-sector projections with weights determined solely by the Abelian component of the group.
This provides a controlled and well-defined notion of channel averaging beyond the unitary setting, where naive uniform averaging is not available.

Finally, we address finite realizations of channel averaging.
We introduce a dual averaging representation in which the twirled channel is expressed as a convex mixture of unitary $1$-design channels acting on invariant sectors.
We further establish a design-like reconstruction result showing that weighted group $t$-designs induce channel $t$-designs for $t = t_{\mathrm{in}} + t_{\mathrm{out}}$, allowing exact reconstruction of the twirled Choi operator from a finite sum.

The results presented here provide a general and constructive description of channel averaging that unifies and extends existing approaches.
They offer both conceptual clarification of the role of symmetry in quantum channels and practical tools for analyzing and implementing averaged channels in a wide range of settings.

\subsection{Notation}
Firstly let us introduce notation inspired by
previous works on averaging of quantum states \cite{Markiewicz23, Markiewicz25}.
In this work we assume that the dimensions of elementary systems constituting
input and output subspaces are equal, however the number of elementary
subsystems can be different for input and output spaces. Hence the input and
output Hilbert spaces are specifed respectively as: $\mathcal H_{\ii}=(\mathbb
C^d)^{\otimes t_{\ii}}$ and $\mathcal H_{\ou}=(\mathbb C^d)^{\otimes t_{\ou}}$.
We consider arbitrary channels $\Phi:\mathcal B(\mathcal H_{\ii})\mapsto
\mathcal B(\mathcal H_{\ou})$ mapping $t_{\ii}$ qu-$d$-its to $t_{\ou}$
qu-$d$-its.

Following \cite{Markiewicz25} let us introduce two arbitrary unitary
representations $\pi_{\operatorname{U}}^{\ii}$ and
$\pi_{\operatorname{U}}^{\ou}$ of a unitary group $\ud$, acting respectively on
$\mathcal B(\mathcal H_{\ii})$ and $\mathcal B(\mathcal H_{\ou})$ via maps:
\begin{eqnarray}
    \label{def:T}
    \mathcal
    T_U^{\ii}(\sigma_{\ii})&=&\pi^{\ii}_{\operatorname{U}}(U)\sigma_{\ii}
    \pi^{\ii}_{\operatorname{U}}(U)^{\dagger},\,\,\sigma_{\ii}\in\mathcal
    B(\mathcal H_{\ii})\nonumber\\
    \mathcal T_U^{\ou}(\sigma_{\ou})&=&\pi^{\ou}_{\operatorname{U}}(U)\sigma_{\ou} \pi^{\ou}_{\operatorname{U}}(U)^{\dagger},\,\,\sigma_{\ou}\in\mathcal B(\mathcal H_{\ou}).
\end{eqnarray}
Clearly we have: $\mathcal
T_{U^{-1}}^{\ii}(\sigma_{\ii})=\pi^{\ii}_{\operatorname{U}}(U)^{\dagger}\sigma_{\ii}
\pi^{\ii}_{\operatorname{U}}(U)$, since
$\pi^{\ii}_{\operatorname{U}}(U^{-1})=\pi^{\ii}_{\operatorname{U}}(U)^{-1}=\pi^{\ii}_{\operatorname{U}}(U)^{\dagger}$.
The last property holds because any representation is a group homomorphism.
Following ideas of \cite{Bartlett07, Miatto12, Graydon22, Grinko2024}, let us define twirling of quantum channels as follows:
\begin{equation}
    \label{def:GU}
    \mathcal M^{\Phi}_{\operatorname{U}}(\rho)=\int_{\ud}\operatorname{d}\!U\,\mathcal T_U^{\ou}\circ \Phi \circ \mathcal T_{U^{-1}}^{\ii}(\rho)=\int_{\ud}\operatorname{d}\!U\,\mathcal T_U^{\ou}\left(\Phi\left(\mathcal T_{U^{-1}}^{\ii}(\rho)\right)\right),
\end{equation}
in which the integration is performed with respect to a Haar measure on the unitary group. Further on we will generalize the above averaging procedure to the case of non-compact transformations.

\subsection{Main results}

Our contributions can be summarized as follows.

\begin{enumerate}
\item In the collective setting, we introduce a partial-transpose reduction that removes the contragredient action and converts the mixed Schur--Weyl (walled Brauer) twirl into an \emph{ordinary} Schur--Weyl twirl on $t_{\ii}+t_{\ou}$ tensor factors, enabling permutation-based formulas without constructing walled Brauer idempotents.
\item We extend channel twirling beyond compact symmetries to reductive (generally non-compact) groups via Cartan decomposition, obtaining a weighted invariant-sector decomposition with weights determined solely by the Abelian Cartan component.
\item We provide two finite realizations of the channel averaging process: a ``dual averaging'' implementation as a convex mixture of unitary-$1$-design channels on invariant sectors, and a $t$-design-like reconstruction result showing that weighted group $t$-designs induce channel $t$-designs for $t=t_{\ii}+t_{\ou}$.
\end{enumerate}

The structural Choi-level identity underpinning channel twirling
(Lemma~\ref{lem:choiform}) records the standard induced action arising from
pre-/post-processing and channel--state duality. Our main technical novelty
begins in the \emph{collective} setting, where Lemma~\ref{lem:UBar} yields a
partial-transpose reduction that removes the contragredient action and converts
the mixed (walled Brauer) twirl into an ordinary Schur--Weyl twirl on
$t_{\ii}+t_{\ou}$ tensor factors. This enables explicit permutation-operator
formulas without constructing walled Brauer idempotents or invoking mixed Schur
transforms.

\paragraph{Partial-transpose reduction (collective setting).}
\begin{theorem}[Collective channel twirl reduces to an ordinary Schur--Weyl twirl, Section \ref{sec:unitary}]
Assume the collective representations $\pi_{\mathrm{U}}^{\ii}(U)=U^{\otimes t_{\ii}}$ and
$\pi_{\mathrm{U}}^{\ou}(U)=U^{\otimes t_{\ou}}$.
Let $\Gamma$ denote partial transpose on the input factor of $\mathcal H_{\ou}\otimes\mathcal H_{\ii}$.
Then the Choi operator of the Haar-twirled channel satisfies
\[
J_{\mathcal M^{\Phi}_{\mathrm U}}
=
\left[
\int_{\ud} \mathrm dU\;
\Big(U^{\otimes (t_{\ii}+t_{\ou})}\Big)\,J_\Phi^\Gamma\,
\Big(U^{\otimes (t_{\ii}+t_{\ou})}\Big)^\dagger
\right]^\Gamma .
\]
Consequently, evaluating $\mathcal M^{\Phi}_{\mathrm U}$ can be reduced to a standard Schur--Weyl projection of $J_\Phi^\Gamma$, so one may work directly with permutation operators on $(t_{\ii}+t_{\ou})$ tensor factors.
\end{theorem}

\paragraph{Beyond compact symmetries: reductive (Cartan) twirl.}
\begin{theorem}[Reductive (Cartan) channel twirl and invariant-sector decomposition, Section \ref{sec:nonunitary}]
Let $\mathcal G$ be a reductive group with Cartan decomposition $\mathcal G = \mathcal K\,\mathcal A\,\mathcal K$,
and let $\pi_{\mathcal G}^{\ii/\ou}$ be finite-dimensional matrix representations on $\mathcal H_{\ii/\ou}$.
Define channel twirling by integrating over $\mathcal K\times\mathcal A\times\mathcal K$ (with normalized Abelian component) and using the adjoint action on the input side.
Then the averaged Choi operator admits an invariant-sector decomposition
\[
J_{\mathcal M^{\Phi}_{\mathcal G}}
=
\sum_{k}\frac{\beta_k}{D_k}\;\mathcal P_k(J_\Phi),
\]
where $\{\mathcal P_k\}$ are commutant-sector projections determined by the restriction of
$\pi_{\mathcal G}^{\ou}\otimes\bar{\pi}_{\mathcal G}^{\ii}$ to the maximal compact subgroup $\mathcal K$, while the coefficients $\beta_k$ depend only on the Abelian component $\mathcal A$ via a real-valued integral.
Moreover, in the collective case one may again express $\mathcal P_k$ using ordinary Schur--Weyl permutation operators after applying the partial transpose trick.
\end{theorem}

\paragraph{Finite realizations: dual averaging and induced channel designs.}
\begin{theorem}[Dual averaging as a mixture of unitary-$1$-design channels, Section \ref{sec:dual}]
The Cartan twirl $\mathcal M_{\mathcal G}^\Phi$ can be represented as a probabilistic mixture of unitary-$1$-design channels acting on invariant sectors:
\[
J_{\mathcal M_{\mathcal G}^\Phi}
=
\sum_k \frac{\beta_k}{D_k}\;\frac{1}{(D_k^G)^2}\sum_{l=1}^{(D_k^G)^2}\tilde{\gamma}^{(k)}_l\,J_\Phi\,\tilde{\gamma}^{(k)\dagger}_l,
\]
where $\{\tilde{\gamma}^{(k)}_l\}$ are Kraus operators obtained by embedding a unitary operator basis on respective irreducible subspaces into the full representation space.
\end{theorem}

\begin{proposition}[Channel $t$-design induced by a weighted group $t$-design, Section \ref{sec:designs}]
Let $t=t_{\ii}+t_{\ou}$ and suppose $\{(G_i,w_i)\}$ is a weighted group $t$-design for a subgroup $\mathcal G\subseteq \mathrm{GL}(d,\mathbb{C})$.
Then the averaged Choi operator can be reconstructed from a finite sum:
\[
J_{\mathcal M_{\mathcal G}^\Phi}
=
\left[
\int_{\mathcal G} \mathrm dG\; G^{\otimes t}\,J_\Phi^\Gamma\,(G^{\otimes t})^\dagger
\right]^\Gamma
=
\left[
\sum_i w_i\; G_i^{\otimes t}\,J_\Phi^\Gamma\,(G_i^{\otimes t})^\dagger
\right]^\Gamma.
\]
\end{proposition}

\paragraph{Warm-up identity: induced Choi action.}
For completeness, we also record the basic Choi-level identity that relates pre/post twirling of channels to the induced action on the Choi operator (Section~\ref{sec:unitary}).
Since variants of this statement appear in several contexts in the literature, we treat it as a structural lemma underpinning the subsequent constructions.

\begin{lem}[Choi form of channel twirling under an induced representation]
\label{lem:choiform}
Let $\Phi:\mathcal{B}(\mathcal H_{\ii})\to\mathcal{B}(\mathcal H_{\ou})$ be a quantum channel with Choi operator $J_\Phi$.
Given unitary representations $\pi_{\mathrm{U}}^{\ii}$ on $\mathcal H_{\ii}$ and $\pi_{\mathrm{U}}^{\ou}$ on $\mathcal H_{\ou}$, define
\[
\mathcal M_{\mathrm U}^{\Phi}(\rho)=\int_{\ud} dU\;
\mathcal T^{\ou}_U\circ \Phi\circ \mathcal T^{\ii}_{U^{-1}}(\rho),
\qquad
\mathcal T^{\ii/\ou}_U(X)=\pi_{\mathrm{U}}^{\ii/\ou}(U)\,X\,\pi_{\mathrm{U}}^{\ii/\ou}(U)^\dagger.
\]
Then
\[
J_{\mathcal M_{\mathrm U}^{\Phi}}
=
\int_{\ud} dU\;
\Big(\pi_{\mathrm{U}}^{\ou}(U)\otimes \bar{\pi}_{\mathrm{U}}^{\ii}(U)\Big)\,J_\Phi\,
\Big(\pi_{\mathrm{U}}^{\ou}(U)\otimes \bar{\pi}_{\mathrm{U}}^{\ii}(U)\Big)^\dagger.
\]
\end{lem}

\section{Averaging quantum channels over arbitrary unitary transformations}
\label{sec:unitary}

In this section we provide an analytic formula for the action of the  channel \eqref{def:GU}
utilizing channel-state duality. Let
$J_{\Phi}=(\Phi\otimes\id_{\ii})\left(\rme\right)\in \mathcal B(\mathcal
H_{\ou})\otimes \mathcal B(\mathcal H_{\ii})$ denote the Choi matrix of the
channel $\Phi$, in which $\rme$ is a non-normalized density matrix of a
maximally entangled state on $\mathcal H_{\ii}\otimes \mathcal H_{\ii}$ and
$\id_{\ii}$ is the identity channel on $\mathcal B(\mathcal H_{\ii})$. Our aim
is to provide an analytic form of the Choi matrix of the channel
$\mathcal{G}^{\Phi}_U$. Utilizing definition of the Choi matrix we obtain:
\begin{eqnarray}
    \label{eq:TwirlChoi}
    J_{\mathcal M^{\Phi}_{\operatorname{U}}}&=&\left(\mathcal
    M^{\Phi}_{\operatorname{U}}\otimes\id_{\ii}\right)\rme\nonumber\\
    &=&\left(\int_{\ud}\operatorname{d}\!U(\mathcal T_U^{\ou}\otimes\id_{\ii}) \circ(\Phi\otimes\id_{\ii})\circ(\mathcal T_{U^{-1}}^{\ii}\otimes\id_{\ii})\right)\rme\nonumber\\
    &=&\int_{\ud}\operatorname{d}\!U \left(\pi_{\operatorname{U}}^{\ou}(U)\otimes\id_{\ii}\right)\left((\Phi\otimes\id_{\ii})\left[\left(\pi_{\operatorname{U}}^{\ii}(U)^{\dagger}\otimes\id_{\ii}\right)\rme\left(\pi_{\operatorname{U}}^{\ii}(U)\otimes\id_{\ii}\right)\right]\right)\left(\pi_{\operatorname{U}}^{\ou}(U)^{\dagger}\otimes\id_{\ii}\right).\nonumber\\
\end{eqnarray}
Now we can extract the inner unitary representation using the following trick,
which holds for arbitrary matrix $X\in\mathcal B(\mathcal H_{\ii})$:
\begin{equation}
    \label{AT}
    (X^{\dagger}\otimes\id_{\ii})\rme (X\otimes\id_{\ii})=(\id_{\ii}\otimes \bar{X})\rme (\id_{\ii}\otimes X^T),
\end{equation}
in which bar denotes complex conjugation.
Applying this property we change the order of the inner unitary representations
in the tensor product:
\begin{eqnarray}
    \label{eq:TwirlChoi2}
    J_{\mathcal M^{\Phi}_{\operatorname{U}}}
    &=&\int_{\ud}\operatorname{d}\!U \left(\pi_{\operatorname{U}}^{\ou}(U)\otimes\id_{\ii}\right)\left((\Phi\otimes\id_{\ii})\left[\left(\id_{\ii}\otimes \Bar{\pi}_{\operatorname{U}}^{\ii}(U)\right)\rme\left(\id_{\ii}\otimes \pi_{\operatorname{U}}^{\ii}(U)^T\right)\right]\right)\left(\pi_{\operatorname{U}}^{\ou}(U)^{\dagger}\otimes\id_{\ii}\right).\nonumber\\
\end{eqnarray}
At this point we need the following lemma:
\begin{lem}
Let $\Phi:\mathcal B(\mathcal H_{\ii})\mapsto \mathcal B(\mathcal H_{\ou})$  and
$\Psi:\mathcal B(\mathcal H_{\ii})\mapsto \mathcal B(\mathcal H_{\ii})$ be two
channels. Then the following commutation relation holds:
\begin{equation}
(\Phi\otimes\id_{\ii})\circ(\id_{\ii}\otimes\Psi)=(\id_{\ou}\otimes\Psi)\circ(\Phi\otimes\id_{\ii}).
\end{equation}
\end{lem}
\noindent Utilizing the above lemma we obtain:
\begin{eqnarray}
    \label{eq:TwirlChoi3}
    J_{\mathcal M^{\Phi}_{\operatorname{U}}}
    &=&\int_{\ud}\operatorname{d}\!U \left(\pi_{\operatorname{U}}^{\ou}(U)\otimes \bar{\pi}_{\operatorname{U}}^{\ii}(U)\right)\left((\Phi\otimes\id_{\ii})\rme\right)\left(\pi_{\operatorname{U}}^{\ou}(U)\otimes \bar{\pi}_{\operatorname{U}}^{\ii}(U)\right)^{\dagger}.\nonumber\\
     &=&\int_{\ud}\operatorname{d}\!U \left(\pi_{\operatorname{U}}^{\ou}(U)\otimes \bar{\pi}_{\operatorname{U}}^{\ii}(U)\right)J_{\Phi}\left(\pi_{\operatorname{U}}^{\ou}(U)\otimes \bar{\pi}_{\operatorname{U}}^{\ii}(U)\right)^{\dagger}.\nonumber\\
\end{eqnarray}
The above integral can be always formally expressed in a closed form using
generalized duality for group representations. Indeed, note that:
\begin{equation}
\label{piU}
\pi_{\operatorname{U}}=\pi_{\operatorname{U}}^{\ou}\otimes \bar{\pi}_{\operatorname{U}}^{\ii},
\end{equation}
is a finite dimensional rational representation of the unitary group on a
representation space $(\mathbb C^d)^{\otimes(t_{\ii}+t_{\ou})}$. This is
because: (i) the contragradient representation \cite{Studzinski17}
$\bar{\pi}_{\operatorname{U}}^{\ii}(U)={\pi}_{\operatorname{U}}^{\ii}(U^{-1})^T$
is a rational representation (\cite{Goodman09}, Sec. 1.5.2), and (ii) tensor
product of two rational representations is a rational representation
(\cite{Goodman09}, \textit{ibid}). Therefore one can apply the generalized
duality theorem for group representations (\cite{Goodman09}, Sec. 4.2.1)\footnote{Note that the assumption of \textit{local regularity} in the general duality theorem is in the finite dimensional case equivalent to rationality of the representation, see Sec. 1.5.1 of \cite{Goodman09}.}, which
boils down to a conclusion, that the action of $\pi_{\operatorname{U}}\times
\pi_{\mathcal C}$, where $\pi_{\mathcal C}$ is a commutant of
$\pi_{\operatorname{U}}$ in the algebra of endomorphisms of $(\mathbb
C^d)^{\otimes(t_{\ii}+t_{\ou})}$, decomposes uniquely into multiplicity-free
direct sum of tensor products of irreducible representations
$\xi_{\operatorname{U}}$ of the unitary group $\operatorname{U}(d)$ and
irreducible representations $\xi_{\mathcal C}$ of a group which is represented
by the commutant $\pi_{\mathcal C}$:
\begin{equation}
\label{eq:isodecomp}
    \pi_{\operatorname{U}}\times\pi_{\mathcal C}=\bigoplus_{k\in\mathcal I} \xi_{\operatorname{U}}^{(k)}\otimes \xi_{\mathcal C}^{(k)}.
\end{equation}
Now, following \cite{Markiewicz25}, one can introduce generalized Schur basis
for the representation space $(\mathbb C^d)^{\otimes(t_{\ii}+t_{\ou})}$, which
has the property that it block-diagonalizes the action of \eqref{eq:isodecomp}.
Schur basis is spanned by vectors of the form $\{\ket{k,m,\lambda}\}$, in which
index $k$ numbers different irreducible subspaces, whereas indices $m$ and
$\lambda$ define internal structure of these subspaces related with irreducible
representations of $\xi_{\operatorname{U}}^{(k)}$ and $\xi_{\mathcal{C}}^{(k)}$. The
transition from standard product basis and the Schur basis
$\{\ket{k,m,\lambda}\}$ is realized by a unitary operation called Quantum Schur
Transform, which is explicitly known only for several concrete representations,
see e.g. \cite{HarrowPHD, SWC_Bacon06,SWC_Kirby18,SWC_Krovi19}. Having defined Schur basis
we can  introduce  the following outer-product-based operators, as shown in
recent works \cite{Markiewicz23, Schlichtholz24}:
    \begin{equation}
    \label{FullPiBasis}
    \hat\Pi_{kk'}^{m_1\lambda_1 m_2\lambda_2}=\ket{k,m_1,\lambda_1}\!\bra{k',m_2,\lambda_2}.
\end{equation}
With the help of the above operators we can define two \textit{Schur operator
bases}:
\begin{eqnarray}
\label{PiBasis0}
\hat\Lambda^{\lambda_1\lambda_2}_k&=&\sum_{m=1}^{D^k_U}\hat\Pi_{kk}^{m\lambda_1 m\lambda_2},\nonumber\\
\hat\Pi^{m_1 m_2}_k&=&\sum_{\lambda=1}^{D^k_C}\hat\Pi_{kk}^{m_1\lambda m_2\lambda},
\end{eqnarray}
in which $D^k_U$ and $D^k_C$ denote ranges of indices respectively $m$ and
$\lambda$. The operators $\hat\Lambda^{\lambda_1\lambda_2}_k$ span operator
representation of  irreducible representations of $\pi_{\mathcal C}$. 
From now on  we will be using the notation $\hat\Lambda^{\lambda_1\lambda_2}_k[\pi_{\operatorname{U}}^{\ou}\otimes
\bar{\pi}_{\operatorname{U}}^{\ii}]$ in order to explicitly indicate the representation from which the operators are derived.
Now, note
that the pair $\{\pi_{\operatorname{U}}, \pi_{\mathcal{C}}\}$ is a dual
reductive pair \cite{Goodman09}, and as shown in \cite{Markiewicz23} in Theorem
2 one can represent the integral \eqref{eq:TwirlChoi3} as a projection of the
input state (here $J_{\Phi}$) to the basis $\hat\Lambda^{\lambda_1\lambda_2}_k[\pi_{\operatorname{U}}^{\ou}\otimes
\bar{\pi}_{\operatorname{U}}^{\ii}]$:
\begin{eqnarray}
    \label{eq:UTwirl}
     J_{\mathcal M^{\Phi}_{\operatorname{U}}}&=&\int_{\ud}\operatorname{d}\!U \left(\pi_{\operatorname{U}}^{\ou}(U)\otimes \bar{\pi}_{\operatorname{U}}^{\ii}(U)\right)J_{\Phi}\left(\pi_{\operatorname{U}}^{\ou}(U)\otimes \bar{\pi}_{\operatorname{U}}^{\ii}(U)\right)^{\dagger}\nonumber\\
    &=&\int_{\ud}\operatorname{d}\!U \pi_{\operatorname{U}}(U)J_{\Phi}\left(\pi_{\operatorname{U}}(U)\right)^{\dagger}\nonumber\\
    &=&\sum_{k\in\mathcal{I}}\frac{1}{D^k_U}\sum_{\lambda_1\lambda_2=1}^{D^k_C}\operatorname{Tr}\left(J_{\Phi}\hat\Lambda_{k}^{\lambda_1\lambda_2\dagger}[\pi_{\operatorname{U}}^{\ou}\otimes
\bar{\pi}_{\operatorname{U}}^{\ii}]\right)\hat\Lambda_{k}^{\lambda_1\lambda_2}[\pi_{\operatorname{U}}^{\ou}\otimes
\bar{\pi}_{\operatorname{U}}^{\ii}].
\end{eqnarray}

In order to utilize this solution one has to explicitly construct operators
$\hat\Lambda^{\lambda_1\lambda_2}_k[\pi_{\operatorname{U}}^{\ou}\otimes
\bar{\pi}_{\operatorname{U}}^{\ii}]$ which span irreducible representations of
the commutant of $\pi_{\operatorname{U}}^{\ou}\otimes
\bar{\pi}_{\operatorname{U}}^{\ii}$. Note that the typical choice of the
representations on input and output ports of the channel is the collective
action of the unitaries on elementary subsystems:
\begin{eqnarray}
\label{eq:collReps}
\pi_{\operatorname{U}}^{\ii}(U)&=&U^{\otimes t_{\ii}},\nonumber\\
\pi_{\operatorname{U}}^{\ou}(U)&=&U^{\otimes t_{\ou}},
\end{eqnarray}
and hence one gets:
\begin{equation}
\label{piUColl}
\pi_{\operatorname{U}}=\pi_{\operatorname{U}}^{\ou}\otimes \bar{\pi}_{\operatorname{U}}^{\ii}=U^{\otimes t_{\ou}}\otimes \bar{U}^{\otimes t_{\ii}}.
\end{equation}
The commutant of this representation is known to be the so-called walled Brauer
algebra \cite{Studzinski17, Mozrzymas18, Studzinski25}, the representations of
which are difficult to construct \cite{Studzinski25}. Although there exist
implementations of Quantum Schur Transform for this case (so called Mixed
Quantum Schur Transform \cite{Nguyen23,Grinko23, Grinko23gt}), one can wish to avoid these
representations and utilize simpler operator representations of permutations. In
order to do so, let us reexpress the averaging operations in a way which removes
the complex conjugation. We start from the following lemma:
\begin{lem}
\label{lem:UBar}
Let $X,Y$ be arbitrary matrices acting on $\mathbb{C}^r$ and let $\sigma$ be any
state on $\mathbb{C}^r\otimes \mathbb{C}^r$. Then the following holds:
\begin{equation}
    \label{eq:lemBar}
    (X\otimes\Bar{Y})\sigma (X\otimes\Bar{Y})^{\dagger}=\left[ (X\otimes{Y})\sigma^{\Gamma}(X\otimes{Y})^{\dagger}\right]^{\Gamma},
\end{equation}
where $^{\Gamma}$ denotes partial transpose on the second factor.
\end{lem}
\begin{proof}
By the linearity argument it suffices to prove the above property for a product state $\sigma=\xi\otimes\tau$.
Then we have:
\begin{eqnarray}
    \left[ (X\otimes{Y})(\xi\otimes\tau)^{\Gamma}(X\otimes{Y})^{\dagger}\right]^{\Gamma}&=&\left[ (X\otimes{Y})(\xi\otimes\tau^T)(X^{\dagger}\otimes{Y}^{\dagger})\right]^{\Gamma}\nonumber\\
    &=&\left[ (X\xi{X^{\dagger}})\otimes(Y\tau^T{Y}^{\dagger})\right]^{\Gamma}\nonumber\\
    &=& (X\xi{X^{\dagger}})\otimes(Y\tau^T{Y}^{\dagger})^{T}\nonumber\\
    &=& (X\xi{X^{\dagger}})\otimes(\bar{Y}\tau{Y}^{T})\nonumber\\
    &=&(X\otimes\bar{Y})(\xi\otimes\tau)(X^{\dagger}\otimes{Y}^{T})\nonumber\\
    &=&(X\otimes\bar{Y})(\xi\otimes\tau)(X\otimes\bar{Y})^{\dagger}.
\end{eqnarray}
\end{proof}
\noindent Utilizing the above lemma for $X=\pi_{\operatorname{U}}^{\ou}(U)$, $Y=\pi_{\operatorname{U}}^{\ii}(U)$ and $\sigma=J_{\Phi}$ we obtain:
\begin{equation}
    \label{eq:TwirlChoiGen}
    J_{\mathcal M^{\Phi}_{\operatorname{U}}}=\int_{\ud}\operatorname{d}\!U \left[\left(\pi_{\operatorname{U}}^{\ou}(U)\otimes \pi_{\operatorname{U}}^{\ii}(U)\right)J^{\Gamma}_{\Phi}\left(\pi_{\operatorname{U}}^{\ou}(U)\otimes \pi_{\operatorname{U}}^{\ii}(U)\right)^{\dagger}\right]^{\Gamma}.
\end{equation}
For the special case of collective unitaries \eqref{eq:collReps} we have:
\begin{eqnarray}
    \label{eq:TwirlChoi4}
    J_{\mathcal M^{\Phi}_{\operatorname{U}}}
     &=&\int_{\ud}\operatorname{d}\!U \left[\left(U^{\otimes t_{\ou}}\otimes {U}^{\otimes t_{\ii}}\right)J_{\Phi}^{\Gamma}\left(U^{\otimes t_{\ou}}\otimes {U}^{\otimes t_{\ii}}\right)^{\dagger}\right]^{\Gamma}\nonumber\\
     &=&\left[\int_{\ud}\operatorname{d}\!U \left(U^{\otimes (t_{\ou}+t_{\ii})}\right)J_{\Phi}^{\Gamma}\left(U^{\otimes (t_{\ou}+t_{\ii})}\right)^{\dagger}\right]^{\Gamma}.\nonumber\\
\end{eqnarray}
Applying the general formula \eqref{eq:UTwirl} we have:
\begin{eqnarray}
    \label{eq:TwirlChoi5}
    J_{\mathcal M^{\Phi}_{\operatorname{U}}}
    &=&\left[\int_{\ud}\operatorname{d}\!U \left(U^{\otimes (t_{\ou}+t_{\ii})}\right)J_{\Phi}^{\Gamma}\left(U^{\otimes (t_{\ou}+t_{\ii})}\right)^{\dagger}\right]^{\Gamma}.\nonumber\\
     &=&\sum_{k}\frac{1}{D^k_U}\sum_{\lambda_1\lambda_2=1}^{D^k_C}\operatorname{Tr}\left(J_{\Phi}^{\Gamma}\hat\Lambda_{k}^{\lambda_1\lambda_2\dagger}[U^{\otimes (t_{\ou}+t_{\ii})}]\right)\left(\hat\Lambda_{k}^{\lambda_1\lambda_2}[U^{\otimes (t_{\ou}+t_{\ii})}]\right)^{\Gamma},\nonumber\\
\end{eqnarray}
where now the operators $\hat\Lambda^{\lambda_1\lambda_2}_k[U^{\otimes (t_{\ou}+t_{\ii})}]$ are operator
representations of permutation operators, and can be constructed directly using
the methodology of Young diagrams, see e.g. Appendix of \cite{Markiewicz23} and
\cite{Tung} for a theoretical exposition, and
\cite{HarrowPHD,SWC_Bacon06,SWC_Kirby18,SWC_Krovi19} for algorithmic
implementations.

\section{Averaging quantum channels over non-unitary operations defined via Cartan decomposition}
\label{sec:nonunitary}
Let us now try to generalize averaging of quantum channels from the unitary
averaging to averaging over finite-dimensional non-unitary representations of a non-compact
symmetry group $\mathcal G$. Such an averaging can be defined via iterated
integral over the Cartan decomposition of $\mathcal G$, which boils down to a
representation of an arbitrary element of $\mathcal G$ as a composition of
operations from the maximally compact component $\mathcal K$ of the group
intertwined by an element from a maximal Abelian subgroup $\mathcal A$:
$G=KAK',\,K,K'\in\mathcal{K},\,A\in\mathcal{A}$. Such idea of non-compact
averaging has been introduced in \cite{Markiewicz21,Markiewicz23} in the context
of averaging quantum states over SLOCC-type operations, and generalized to
averaging over arbitrary finite-dimensional matrix representations of reductive
groups in the recent work \cite{Markiewicz25}. Following \cite{Markiewicz25},
let us take arbitrary matrix representation $\pi_{\mathcal G}$ of the group
$\mathcal G$ and define the averaging of a quantum state with respect to the
chosen representation as follows:
\begin{equation}
    \mathcal M_{\mathcal{G}}(\rho)=\int_{\mathcal K\times\mathcal A\times\mathcal K} \left(\pi_{\mathcal G}(K)\pi_{\mathcal G}(A_{\textrm{n}})\pi_{\mathcal G}(K')\right)\rho\left(\pi_{\mathcal G}(K)\pi_{\mathcal G}(A_{\textrm{n}})\pi_{\mathcal G}(K')\right)^{\dagger} \operatorname{d}\!K\operatorname{d}\!A \operatorname{d}\!K'.
    \label{mainTwirlG0}
\end{equation}
in which the integration is defined over the Cartan ``KAK'' decomposition of
$\mathcal G$. The integration over the two copies of maximally compact
components $\mathcal K, \mathcal K'$ is performed uniformly, which means that
$\int_{\mathcal K}\operatorname{d}\!K$ represents Haar integral over $\mathcal
K$. Integration over $\mathcal A$ is performed with respect to any finite
normalized measure on this (in general non-compact) manifold, but there appears
one crucial assumption: elements of $\mathcal A$ enter the integral as
normalized matrices: $A_{\textrm{n}}=A/||A||$, in order to guarantee trace
nonincreasing character of the map \eqref{mainTwirlG0}.

Our aim is to define twirling of quantum channels with respect to the representation $\pi_{\mathcal G}$  applied as pre- and post-processing operation:
 \begin{eqnarray}
    \label{def:G}
    \mathcal T_G^{\ii}(\sigma_{\ii})&=&\pi^{\ii}_{{\mathcal G}}(G)\sigma_{\ii} \pi^{\ii}_{{\mathcal G}}(G)^{\dagger},\,\,\sigma_{\ii}\in\mathcal B(\mathcal H_{\ii})\nonumber\\
    \mathcal T_G^{\ou}(\sigma_{\ou})&=&\pi^{\ou}_{{\mathcal G}}(G)\sigma_{\ou} \pi^{\ou}_{{\mathcal G}}(G)^{\dagger},\,\,\sigma_{\ou}\in\mathcal B(\mathcal H_{\ou}).
\end{eqnarray}
In order to follow the analogy with \eqref{def:GU} we have  to define the
counterpart of the $\mathcal T_{U^{-1}}$ operation from definition
\eqref{def:GU}. Formally every element $G$ from the group $\mathcal G$ is
invertible, however, the demand of trace-nonincreasing character of a quantum
channel forces us to integrate over normalized elements from the group $\mathcal
G$ in the twirling channel \eqref{mainTwirlG0}. Now, in general normalized
matrices, which represent trace-nonincreasing quantum channels, are not
invertible as quantum operations. Therefore we define twirling of quantum
channels with respect to $\mathcal G$ operations as follows:
\begin{equation}
    \label{def:GSL}
    \mathcal M^{\Phi}_{\mathcal G}(\rho)=\int\operatorname{d}\!G\,\mathcal T_{G}^{\ou} \circ\Phi\circ\mathcal T_{G}^{\ii\dagger}(\rho),
\end{equation}
in which we define:
\begin{eqnarray}
    \mathcal T_G^{\ii\dagger}(\sigma_{\ii})&=&\pi^{\ii}_{{\mathcal G}}(G)^{\dagger}\sigma_{\ii} \pi^{\ii}_{{\mathcal G}}(G).
    \label{TSLDag}
\end{eqnarray}
Note that the integration in \eqref{def:GSL} is understood in the sense of
Cartan integral \eqref{mainTwirlG0}. In the case of a compact group $\mathcal G$
with trivial $\mathcal A$ this definition is fully equivalent to \eqref{def:GU}.

In analogy with \eqref{eq:TwirlChoi} the Choi matrix of a channel twirled with respect to a group $\mathcal G$ reads:
\begin{eqnarray}
    \label{eq:TwirlChoiG}
    J_{\mathcal M^{\Phi}_{\mathcal G}}&=&\left(\mathcal
    M^{\Phi}_{\mathcal G}\otimes\id_{\ii}\right)\rme\nonumber\\
    &=&\left(\int\operatorname{d}\!G(\mathcal T_G^{\ou}\otimes\id_{\ii}) \circ(\Phi\otimes\id_{\ii})\circ(\mathcal T_{G}^{\ii \dagger}\otimes\id_{\ii})\right)\rme\nonumber\\
    &=&\int\operatorname{d}\!G \left(\pi_{\mathcal G}^{\ou}(G)\otimes\id_{\ii}\right)\left((\Phi\otimes\id_{\ii})\left[\left(\pi_{\mathcal G}^{\ii}(G)^{\dagger}\otimes\id_{\ii}\right)\rme\left(\pi_{\mathcal G}^{\ii}(G)\otimes\id_{\ii}\right)\right]\right)\left(\pi_{\mathcal G}^{\ou}(G)^{\dagger}\otimes\id_{\ii}\right).\nonumber\\
\end{eqnarray}
Now all of the derivations \eqref{AT}--\eqref{eq:TwirlChoi3} concerning
the unitary case hold, since they do not utilize the property of unitarity of the underlying representations, and we have the following expression for the Choi matrix
of the channel \eqref{def:GSL}:
\begin{eqnarray}
    \label{eq:TwirlChoiG0}
    J_{\mathcal M^{\Phi}_{\mathcal G}}
     &=&\int\operatorname{d}\!G \left(\pi_{\mathcal G}^{\ou}(G)\otimes \bar{\pi}_{\mathcal G}^{\ii}(G)\right)J_{\Phi}\left(\pi_{\mathcal G}^{\ou}(G)\otimes \bar{\pi}_{\mathcal G}^{\ii}(G)\right)^{\dagger}.
\end{eqnarray}
Due to distributivity of the tensor product with respect to ordinary matrix
product, we have the following Cartan decompositon:
\begin{equation}
\pi_{\mathcal G}^{\ou}(G)\otimes \bar{\pi}_{\mathcal G}^{\ii}(G)=\left(\pi_{\mathcal G}^{\ou}(K)\otimes \bar{\pi}_{\mathcal G}^{\ii}(K)\right)\left(\pi_{\mathcal G}^{\ou}(A_{\textrm{n}})\otimes \pi_{\mathcal G}^{\ii}(A_{\textrm{n}})\right)\left(\pi_{\mathcal G}^{\ou}(K')\otimes \bar{\pi}_{\mathcal G}^{\ii}(K')\right),
\end{equation}
in which we omit bar in the middle term since $A_{\textrm{n}}$ is a real
diagonal matrix. Since Cartan decomposition commutes with taking representations
of the components \cite{Markiewicz25}, it can be seen that $\pi_{\mathcal
G}^{\ou}\otimes \bar{\pi}_{\mathcal G}^{\ii}|_{\mathcal K}\equiv\pi_{\mathcal
K}$, when acting solely on compact elements $K\in\mathcal{K}$ of the group
$\mathcal{G}$, is itself a representation of the maximal compact subgroup
$\mathcal K$ of $\mathcal G$. Analogously, in the same sense $\pi_{\mathcal
G}^{\ou}\otimes \bar{\pi}_{\mathcal G}^{\ii}|_{\mathcal A}\equiv\pi_{\mathcal
A}$, when acting solely on non-compact elements $A\in\mathcal{A}$ can be treated
as a representation of the maximal Abelian subgroup $\mathcal A$ of $\mathcal
G$. Hence the integral \eqref{eq:TwirlChoiG0} becomes:
\begin{eqnarray}
    \label{eq:TwirlChoiG1}
    J_{\mathcal M^{\Phi}_{\mathcal G}}
     &=&\int_{\mathcal K\times\mathcal A\times\mathcal K} \left(\pi_{\mathcal K}(K)\pi_{\mathcal A}(A_{\textrm{n}})\pi_{\mathcal K}(K')\right)J_{\Phi}\left(\pi_{\mathcal K}(K)\pi_{\mathcal A}(A_{\textrm{n}})\pi_{\mathcal K}(K')\right)^{\dagger} \operatorname{d}\!K\operatorname{d}\!A \operatorname{d}\!K'.
\end{eqnarray}
As shown in \cite{Markiewicz23} such integral decomposes into a product of
compact twirling of the input operator and of the square of the non-compact
part:
\begin{equation}
\mathcal M_{\mathcal G}(\rho)=\mathcal M_{\mathcal K}(\rho)\mathcal M_{\mathcal K}\left(\pi_{\mathcal A}(A_{\textrm{n}})^2\right),
\end{equation}
which can be directly represented as:
\begin{equation}
\label{eq:SLTwirlGen}
    \mathcal M_{\mathcal G}(\rho)=\sum_k\left(\frac{\beta_k}{D^k}\right)\mathcal M_{\mathcal K}^{(k)}(\rho)=\sum_{k}\frac{1}{D^k_U}\left(\frac{\beta_k}{D^k}\right)\sum_{\lambda_1\lambda_2=1}^{D^k_C}\operatorname{Tr}\left(\rho\hat\Lambda_{k}^{\lambda_1\lambda_2\dagger}[\pi_{\mathcal
K}]\right)\hat\Lambda_{k}^{\lambda_1\lambda_2}[\pi_{\mathcal
K}],
\end{equation}
where the coefficients $\beta_k$ are integrals over the non-compact components:
\begin{equation}
\label{betafin}
\beta_k=\Tr\left(\left[\int\left(\pi_{\mathcal A}(A_{\textrm{n}})\right)^2\operatorname{d}\!A\right]\hat\Pi_k^\dagger\right),
\end{equation}
and the operators $\hat\Lambda_{k}^{\lambda_1\lambda_2}[\pi_{\mathcal
K}]$ are Schur operator
basis found by applying generalized duality theorem to the representation
$\pi_{\mathcal K}=\pi_{\mathcal G}^{\ou}\otimes\bar{\pi}_{\mathcal
G}^{\ii}|_{\mathcal K}$. Therefore, we finally have:
\begin{equation}
    \label{eq:finJMPhiG}
J_{\mathcal M^{\Phi}_{\mathcal G}}=\sum_{k}\frac{1}{D^k_U}\left(\frac{\beta_k}{D^k}\right)\sum_{\lambda_1\lambda_2=1}^{D^k_C}\operatorname{Tr}\left(J_{\Phi}\hat\Lambda_{k}^{\lambda_1\lambda_2\dagger}[\pi_{\mathcal
K}]\right)\hat\Lambda_{k}^{\lambda_1\lambda_2}[\pi_{\mathcal
K}].
\end{equation}
In full analogy with the previous case one may want to avoid using
contragradient representations in the generalized Schur-Weyl duality\footnote{Note that although formally $\bar\pi_{\mathcal G}^{\ii}$ is \textit{not} a contragradient representation to $\pi_{\mathcal G}^{\ii}$, for a non-unitary representation $\pi_{\mathcal G}^{\ii}$, the restriction  $\bar\pi_{\mathcal K}^{\ii}$ is contragradient to $\pi_{\mathcal K}^{\ii}$, since we can without loss of generality always take unitary $\pi_{\mathcal K}^{\ii}$.}. Utilizing
Lemma \ref{lem:UBar} the formula \eqref{eq:TwirlChoiG0} becomes:
\begin{eqnarray}
    \label{eq:TwirlChoiG0PT}
    J_{\mathcal M^{\Phi}_{\mathcal G}}
     &=&\int\operatorname{d}\!G \left[\left(\pi_{\mathcal G}^{\ou}(G)\otimes \pi_{\mathcal G}^{\ii}(G)\right)J_{\Phi}^{\Gamma}\left(\pi_{\mathcal G}^{\ou}(G)\otimes \pi_{\mathcal G}^{\ii}(G)\right)^{\dagger}\right]^{\Gamma},
\end{eqnarray}
which can be analogously, using Cartan decomposition, decomposed into:
\begin{equation}
    \label{eq:finJMPhiGPT}
J_{\mathcal M^{\Phi}_{\mathcal G}}=\sum_{k}\frac{1}{D^k_U}\left(\frac{\beta_k}{D^k}\right)\sum_{\lambda_1\lambda_2=1}^{D^k_C}\operatorname{Tr}\left(J_{\Phi}^{\Gamma}\hat\Lambda_{k}^{\lambda_1\lambda_2\dagger}[\pi_{\mathcal K}^{\ou}\otimes\pi_{\mathcal{K}}^{\ii}]\right)\left(\hat\Lambda_{k}^{\lambda_1\lambda_2}[\pi_{\mathcal K}^{\ou}\otimes\pi_{\mathcal{K}}^{\ii}]\right)^{\Gamma},
\end{equation}
where now the operators $\hat\Lambda_{k}^{\lambda_1\lambda_2}[\pi_{\mathcal K}^{\ou}\otimes\pi_{\mathcal{K}}^{\ii}]$ are elements of
the Schur operator basis related with decomposing the representations of
$\pi_{\mathcal G}^{\ou}\otimes\pi_{\mathcal{G}}^{\ii}|_{\mathcal{K}}=\pi_{\mathcal K}^{\ou}\otimes\pi_{\mathcal{K}}^{\ii}$, which
are typically much easier to construct. Indeed, let us take as an example the
case of $\mathcal G=\textrm{SL}(d,\mathbb C)$, for which the maximally compact
Cartan component $\mathcal K$ is the special unitary group $\textrm{SU}(d)$. In
full analogy to the unitary case we assume collective action of group elements
on input and output spaces:
\begin{eqnarray}
\pi_{\textrm{SL}}^{\ii}(L)&=&L^{\otimes t_{\ii}},\,\,L\in\textrm{SL}(d,\mathbb C),\nonumber\\
\pi_{\textrm{SL}}^{\ou}(L)&=&L^{\otimes t_{\ou}},\,\,L\in\textrm{SL}(d,\mathbb C).
\end{eqnarray}
In this case the Choi matrix of a twirled channel reads after
\eqref{eq:TwirlChoiG0PT}:
\begin{eqnarray}
    \label{eq:TwirlChoiG0SLT}
    J_{\mathcal M^{\Phi}_{\textrm{SL}}}
     &=&\int_{\textrm{SL}(d,\mathbb C)}\operatorname{d}\!L \left[\left(L^{\otimes(t_{\ii}+t_{\ou})}\right)J_{\Phi}^{\Gamma}\left(L^{\otimes(t_{\ii}+t_{\ou})}\right)^{\dagger}\right]^{\Gamma}.
\end{eqnarray}
Now the restricted representation $\pi_{\mathcal
G}^{\ou}\otimes \pi_{\mathcal G}^{\ii}|_{\mathcal K}$ reads in this case simply 
$U^{\otimes(t_{\ii}+t_{\ou})}$, for $U\in\textrm{SU}(d)$, and we
have, following \eqref{eq:finJMPhiGPT}: 
\begin{equation}
    \label{eq:finJMPhiSLPT}
J_{\mathcal M^{\Phi}_{\textrm{SL}}}=\sum_{k}\frac{1}{D^k_U}\left(\frac{\beta_k}{D^k}\right)\sum_{\lambda_1\lambda_2=1}^{D^k_C}\operatorname{Tr}\left(J_{\Phi}^{\Gamma}\hat\Lambda_{k}^{\lambda_1\lambda_2\dagger}[U^{\otimes(t_{\ii}+t_{\ou})}]\right)\left(\hat\Lambda_{k}^{\lambda_1\lambda_2}[U^{\otimes(t_{\ii}+t_{\ou})}]\right)^{\Gamma},
\end{equation}
where now $\hat\Lambda_{k}^{\lambda_1\lambda_2}[U^{\otimes(t_{\ii}+t_{\ou})}]$ are ordinary permutation
operators corresponding to standard Schur-Weyl duality for
$\textrm{U}^{\otimes(t_{\ii}+t_{\ou})}$ and the permutation group $\textrm{S}_{(t_{\ii}+t_{\ou})}$.

\section{Dual averaging of channels}
\label{sec:dual}
In this section we utilize the concept of \textit{generalized finite unitary
averaging} introduced in a recent work \cite{Markiewicz25} in the context of
averaging of quantum channels. The idea of generalized finite unitary averaging,
originally introduced for averaging quantum states over general symmetry
operations, also bases on generalized duality for group representations,
however, in analogy to the concept of unitary $t$-designs, it keeps the
averaging procedure on the original group representation side, and not on its
commutant, as in formulas \eqref{eq:UTwirl} and \eqref{eq:finJMPhiG}, in which
the Choi matrix of the twirled channel is expressed in terms of operators
$\hat\Lambda_{k}^{\lambda_1\lambda_2}$, which span the commutant of tensor
product of the original symmetry group. Indeed, the idea in \cite{Markiewicz25}
is to represent the Cartan-decomposition-based average of a quantum state with respect to arbitrary symmetry
group \eqref{mainTwirlG0} as a probabilistic mixture of unitary-$1$-designs on
invariant subspaces:
\begin{equation}
    \label{def-fin-group-average}
    \mathcal M_{\mathcal{G}}(\rho)=\sum_{k\in\mathcal{I}}p_k\mathcal{D}_{\mathcal K}^{(k)}(\rho),
\end{equation}
in which $\mathcal{D}_{\mathcal K}^{(k)}$ represents unitary-$1$-design on
$D^k_G$-dimensional subspace irreducible under the action of the restricted representation $\pi_{\mathcal K}$, for $\mathcal K$
being the maximal compact subgroup of $\mathcal G$:
\begin{equation}
    \label{def-fin-group-average-mixed-unitary}
   \mathcal{D}_{\mathcal K}^{(k)}(\rho)=\frac{1}{\left(D^k_G\right)^2} \sum_{l=1}^{\left(D^k_G\right)^2} \tilde\gamma_l^{(k)}\rho\,\tilde\gamma_l^{(k)\dagger},
\end{equation}
and the probabilities $p_k$ are specified in full analogy to
\eqref{eq:SLTwirlGen} as $p_k=\frac{\beta_k}{D^k}$. Unitary-$1$-designs are
specified by choosing a unitary operator basis
$\{\gamma^{(k)}_l\}_{l=1}^{\left(D^k_G\right)^2}$ on space $\mathbb C^{D^k_G}$,
and further they are extended to act on the entire representation space using
basis operators $\hat\Pi_k^{m_1m_2}[\pi_{\mathcal K}]$ \eqref{PiBasis0} specified by generalized
group duality:
\begin{equation}
    \label{Gt0}
\tilde\gamma^{(k)}_l=\sum_{m_1,m_2=1}^{D^k_G}{\left(\gamma^{(k)}_l\right)}_{m_1m_2}\hat\Pi^{m_1 m_2}_k[\pi_{\mathcal K}].
\end{equation}
Utilizing the above formalism one can express the Choi matrix $J_{\mathcal
M^{\Phi}_{\mathcal G}}$ of a twirled channel \eqref{eq:TwirlChoiG0} in a
design-like form:
\begin{equation}
    \label{eq:finJMPhiG_design}
J_{\mathcal M^{\Phi}_{\mathcal G}}=\sum_{k\in\mathcal I}\left(\frac{\beta_k}{D^k}\right)\mathcal{D}_{\mathcal K}^{(k)}(J_{\Phi})
=\sum_{k\in\mathcal I}\left(\frac{\beta_k}{D^k}\right)\frac{1}{\left(D^k_G\right)^2} \sum_{l=1}^{\left(D^k_G\right)^2} \tilde\gamma_l^{(k)}J_{\Phi}\,\tilde\gamma_l^{(k)\dagger},
\end{equation}
where the operators \eqref{Gt0} are defined for $\hat\Pi^{m_1 m_2}_k[\pi_{\mathcal K}^{\ou}\otimes\bar{\pi}_{\mathcal
K}^{\ii}]$.
In principle in the above formulas one can use arbitrary unitary operator basis,
 however the typical choice due to its simplicity of construction is the
 Heisenberg-Weyl unitary operator bases \cite{BZ17}, comprising operators:
 $$\{\omega_{D^k_G}^{ij}Z_{D^k_G}^iX_{D^k_G}^j\}_{i,j=0}^{{D^k_G}-1}$$ in
 which the complex root of unity reads: $\omega_{{D^k_G}}=e^{2i\pi/{D^k_G}}$,
 and the operators $Z,X$ are unitary generalizations of Pauli matrices
 $\sigma_z,\sigma_x$.
 
\section{Channel $t$-designs}
\label{sec:designs}
The methods of averaging quantum channels described in previous sections are
based on generalized group duality and therefore, though finite, they demand
performing projections on irreducible subspaces related with representations of
either the original symmetry group for input and output ports or the commutant
of its action. One may expect a finite averaging method which respects the
operator structure of the original group actions, as in the case of unitary
$t$-designs. We will provide such a generalization utilizing Lemma
\ref{lem:UBar} for the following case:
\begin{eqnarray}
\pi_{\mathcal{G}}^{\ii}(G)&=&G^{\otimes t_{\ii}},\,\,G\in\mathcal{G},\nonumber\\
\pi_{\mathcal{G}}^{\ou}(G)&=&G^{\otimes t_{\ou}},\,\,G\in\mathcal{G},
\end{eqnarray}
in which $\mathcal G$ is an arbitrary subgroup of $\textrm{GL}(d,\mathbb C)$.
We utilize the following lemma:
\begin{lem}
    \label{lem:Gdesigns}
    Let $\{G_i,w_i\}$ be a (weighted) group $t$-design, namely:
    \begin{equation}
        \label{GdesingDef}
        \int_{\mathcal G}\operatorname{d}\!G\, G^{\otimes t}\rho\, G^{\otimes t\dagger}
        =\sum_{i\in\mathcal X}w_i G_i^{\otimes t}\rho\, G_i^{\otimes t\dagger}.
    \end{equation}
    Then the following holds:
    \begin{equation}
        \label{GdesingGam}
        \int_{\mathcal G}\operatorname{d}\!G\, \left[G^{\otimes t}\rho^{\Gamma}\, G^{\otimes t\dagger}\right]^{\Gamma}
        =\sum_{i\in\mathcal X}w_i \left[G_i^{\otimes t}\rho^{\Gamma}\, G_i^{\otimes t\dagger}\right]^{\Gamma}.
    \end{equation}
\end{lem}
\begin{proof}
The above holds due to: (i) linearity of partial transpose, (ii) the fact that
\eqref{GdesingDef} holds for any $\rho$ (in fact $\rho$ is redundant in this
formula).
\end{proof}
Utilizing the above lemma we can define channel $t$-design, where
$t=t_{\ii}+t_{\ou}$ as a finite set of matrices and weights $\{G_i,w_i\}$ which
reconstruct the Choi matrix of an averaged channel:
\begin{eqnarray}
    \label{eq:DesignChoiG0SLT}
    J_{\mathcal M^{\Phi}_{\mathcal{G}}}
     &=&\int_{\mathcal{G}}\operatorname{d}\!G \left[\left(G^{\otimes(t_{\ii}+t_{\ou})}\right)J_{\Phi}^{\Gamma}\left(G^{\otimes(t_{\ii}+t_{\ou})}\right)^{\dagger}\right]^{\Gamma}\nonumber\\
     &=&\sum_{i\in\mathcal X}w_i\left[\left(G_i^{\otimes(t_{\ii}+t_{\ou})}\right)J_{\Phi}^{\Gamma}\left(G_i^{\otimes(t_{\ii}+t_{\ou})}\right)^{\dagger}\right]^{\Gamma}.
\end{eqnarray}
The above construction is very general and can be applied whenever one can construct $t$-design for a given symmetry group $\mathcal G$ and 
demanded value of $t=t_{\ii}+t_{\ou}$. In the case of $\mathcal G$ being a unitary group 
there is a huge literature on this topic \cite{Dankert05, Dankert09, Gross07, Roy09, Webb16, Nakata21}, whereas in the case of $\mathcal G=\textrm{SL}(2,\mathbb C)$ 
one can construct respective designs as product designs comprising unitary ones and Gauss-type quadratues, see \cite{Markiewicz21} for detailed presentation.

\paragraph{Relation to designs for channels.}
Several notions of designs for channels have been proposed in the literature,
often by pushing forward a design on a larger unitary group through a fixed
Stinespring dilation or by imposing moment-matching conditions on ensembles of
CPTP maps, see e.g. \cite{Czartowski25}. The channel $t$-design notion in Eq.~(43) is tailored to the present
twirling problem: it is the minimal finite condition required to reproduce the
averaged Choi operator for collective pre/post actions, and it follows directly
from the partial-transpose  trick and any weighted group $t$-design for
$t=t_{\rm in}+t_{\rm out}$. This perspective complements other channel-design
frameworks by making the symmetry and the induced representation on the Choi
space explicit.

\section{Conclusions and outlook}

We have shown that channel twirling is most naturally and generally formulated at the level of Choi operators, where it becomes a genuine group twirl under the induced output--input representation. From this perspective, channel averaging is no longer merely an operational pre- and post-processing prescription, but a representation-theoretic projection acting directly on the Choi operator space. This shift of viewpoint provides a unified and explicit description of channel twirling that applies to arbitrary input and output representations and does not rely on case-by-case constructions.

While symmetry-restricted quantum channels---covariant or equivariant---have been extensively studied, existing works typically focus on characterizing the set of invariant channels rather than on the twirling map itself. In contrast, the present work treats the twirl as an explicit linear superoperator acting on Choi operators. Concretely, the Haar (or Cartan) average defines a completely positive linear map whose range coincides with the commutant of the induced representation, so that the twirl realizes a projection onto this invariant subspace in the operator-algebraic sense. 
This distinction is crucial: instead of merely identifying the structural constraints that invariant channels must satisfy, we obtain a constructive formula for the averaging map itself, allowing one to compute the averaged channel directly rather than inferring its form indirectly from symmetry considerations.

In the collective setting, where the symmetry acts as tensor powers on multiple input and output subsystems, the relevant commutant is the walled Brauer algebra. A central insight of this work is that the technical complexity associated with mixed Schur--Weyl duality can be bypassed by a partial-transpose reduction: the channel twirl is mapped to an ordinary Schur--Weyl twirl of the partially transposed Choi operator under a single tensor-power representation. This reduction replaces contragredient actions by a final partial transpose and enables practical evaluation using standard permutation-operator techniques, without constructing explicit walled Brauer idempotents or mixed Schur transforms.

We further extended the Choi-level formulation of channel twirling beyond compact symmetries acting via unitary representations to reductive, generally non-compact groups via Cartan decomposition. In this setting, where uniform averaging is not available, the twirl decomposes into invariant-sector projections determined by the maximal compact subgroup, weighted by coefficients that depend solely on the Abelian Cartan component. To our knowledge, this provides the first systematic and constructive treatment of channel twirling for non-compact transformations at the Choi level, clarifying both its algebraic structure and its normalization properties.

Two finite realizations of channel averaging were also developed. First, we showed that the Cartan twirl admits a dual representation as a probabilistic mixture of unitary-$1$-design channels acting on invariant sectors through embedded operator bases. Second, motivated by unitary designs, we introduced a notion of channel $t$-design tailored to the present setting: any weighted group $t$-design with $t=t_{\mathrm{in}}+t_{\mathrm{out}}$ induces an exact finite reconstruction of the averaged Choi operator. Together, these constructions demonstrate that the abstract Choi-level twirl admits finite, symmetry-respecting implementations without leaving the representation-theoretic framework.

Several directions follow naturally from this work. On the constructive side, it
is important to develop efficient exact and approximate finite averaging sets
for broader classes of symmetry groups and representations. A key open problem
is to quantify the approximation error incurred when replacing Haar or Cartan
twirling by a finite weighted sum (design-like implementation). Concretely, one would like bounds on the deviation of the
resulting twirling superoperator from the ideal one, either at the Choi level
(e.g.\ in trace or Hilbert--Schmidt norm for Choi operators) or directly at the
channel level (e.g.\ in diamond norm), with explicit dependence on
$t_{\mathrm{in}}+t_{\mathrm{out}}$, the local dimension, and the quality/size of
the design.

From an algorithmic perspective, the partial-transpose reduction suggests
scalable implementations of channel twirls based on permutation operators and
Schur--Weyl projectors, with complexity controlled by $t_{\mathrm{in}}+t_{\mathrm{out}}$
rather than by the full mixed Schur--Weyl machinery. More broadly, treating
channel symmetries directly at the Choi level appears essential for extending
symmetry-based reductions to approximate symmetries, non-identical local
dimensions, and learning problems involving structured quantum noise.

\section*{Acknowledgements}
MM acknowledges support from the National Science Center (NCN), Poland, under
Project Opus No. 2024/53/B/ST2/02026. ZP and ŁP acknowledge support from the
National Science Center (NCN), Poland, under Project Opus No.
2022/47/B/ST6/02380.

\bibliographystyle{quantum}

\begin{thebibliography}{10}

\bibitem{Bartlett03a}
Stephen~D. Bartlett, Terry Rudolph, and Robert~W. Spekkens.
\newblock ``Classical and quantum communication without a shared reference frame''.
\newblock \href{https://dx.doi.org/10.1103/PhysRevLett.91.027901}{Phys. Rev. Lett. {\bf 91}, 027901}~(2003).

\bibitem{Bartlett07}
Stephen~D. Bartlett, Terry Rudolph, and Robert~W. Spekkens.
\newblock ``Reference frames, superselection rules, and quantum information''.
\newblock \href{https://dx.doi.org/10.1103/RevModPhys.79.555}{Rev. Mod. Phys. {\bf 79}, 555--609}~(2007).

\bibitem{Zanardi97}
P.~Zanardi and M.~Rasetti.
\newblock ``Noiseless quantum codes''.
\newblock \href{https://dx.doi.org/10.1103/PhysRevLett.79.3306}{Phys. Rev. Lett. {\bf 79}, 3306--3309}~(1997).

\bibitem{Eggeling01}
T.~Eggeling and R.~F. Werner.
\newblock ``{Separability properties of tripartite states with $U\ensuremath{\otimes}U\ensuremath{\otimes}U$ symmetry}''.
\newblock \href{https://dx.doi.org/10.1103/PhysRevA.63.042111}{Phys. Rev. A {\bf 63}, 042111}~(2001).

\bibitem{Bae19}
Joonwoo Bae, Beatrix~C Hiesmayr, and Daniel McNulty.
\newblock ``Linking entanglement detection and state tomography via quantum 2-designs''.
\newblock \href{https://dx.doi.org/10.1088/1367-2630/aaf8cf}{New Journal of Physics {\bf 21}, 013012}~(2019).

\bibitem{Gross21}
David Gross, Sepehr Nezami, and Michael Walter.
\newblock ``{Schur--Weyl} duality for the {Clifford} group with applications: Property testing, a robust {Hudson} theorem, and de {Finetti} representations''.
\newblock \href{https://dx.doi.org/10.1007/s00220-021-04118-7}{Communications in Mathematical Physics {\bf 385}, 1325--1393}~(2021).

\bibitem{Nakata21}
Yoshifumi Nakata, Da~Zhao, Takayuki Okuda, Eiichi Bannai, Yasunari Suzuki, Shiro Tamiya, Kentaro Heya, Zhiguang Yan, Kun Zuo, Shuhei Tamate, Yutaka Tabuchi, and Yasunobu Nakamura.
\newblock ``Quantum circuits for exact unitary $t$-designs and applications to higher-order randomized benchmarking''.
\newblock \href{https://dx.doi.org/10.1103/PRXQuantum.2.030339}{PRX Quantum {\bf 2}, 030339}~(2021).

\bibitem{Markiewicz23}
Marcin Markiewicz and Janusz Przewocki.
\newblock ``Duality of averaging of quantum states over arbitrary symmetry groups revealing {Schur–Weyl} duality''.
\newblock \href{https://dx.doi.org/10.1088/1751-8121/acf4d5}{Journal of Physics A: Mathematical and Theoretical {\bf 56}, 395301}~(2023).

\bibitem{Miatto12}
Filippo~M. Miatto.
\newblock ``Dealing with unknown quantum operations''~(2012).
\newblock  \href{http://arxiv.org/abs/1209.4281}{arXiv:1209.4281}.

\bibitem{Winter21}
Martina Gschwendtner, Andreas Bluhm, and Andreas Winter.
\newblock ``Programmability of covariant quantum channels''.
\newblock \href{https://dx.doi.org/10.22331/q-2021-06-29-488}{{Quantum} {\bf 5}, 488}~(2021).

\bibitem{Graydon22}
Matthew~A. Graydon, Joshua Skanes-Norman, and Joel~J. Wallman.
\newblock ``Designing stochastic channels''~(2022).
\newblock  \href{http://arxiv.org/abs/2201.07156}{arXiv:2201.07156}.

\bibitem{Kong22}
Linghang Kong and Zi-Wen Liu.
\newblock ``Near-optimal covariant quantum error-correcting codes from random unitaries with symmetries''.
\newblock \href{https://dx.doi.org/10.1103/PRXQuantum.3.020314}{PRX Quantum {\bf 3}, 020314}~(2022).

\bibitem{Nechita25}
Ion Nechita and Sang-Jun Park.
\newblock ``Random covariant quantum channels''.
\newblock \href{https://dx.doi.org/10.1007/s00023-025-01558-y}{Annales Henri Poincar{\'e}}~(2025).

\bibitem{puchala2011experimentally}
Zbigniew Pucha{\l}a, Jaros{\l}aw~Adam Miszczak, Piotr Gawron, and Bart{\l}omiej Gardas.
\newblock ``Experimentally feasible measures of distance between quantum operations''.
\newblock \href{https://dx.doi.org/10.1007/s11128-010-0186-x}{Quantum Information Processing {\bf 10}, 1--12}~(2011).

\bibitem{krawiec2020discrimination}
Aleksandra Krawiec, {\L}ukasz Pawela, and Zbigniew Pucha{\l}a.
\newblock ``Discrimination of povms with rank-one effects''.
\newblock \href{https://dx.doi.org/10.1007/s11128-020-02925-8}{Quantum Information Processing {\bf 19}, 428}~(2020).

\bibitem{Dipper08}
Richard Dipper, Stephen Doty, and Jun Hu.
\newblock ``{Brauer} algebras, symplectic {Schur} algebras and {Schur-Weyl} duality''.
\newblock Transactions of the American Mathematical Society {\bf 360}, 189--213~(2008).
\newblock  url:~\url{http://www.jstor.org/stable/20161873}.

\bibitem{Studzinski17}
Marek Mozrzymas, Michał Studziński, and Nilanjana Datta.
\newblock ``Structure of irreducibly covariant quantum channels for finite groups''.
\newblock \href{https://dx.doi.org/10.1063/1.4983710}{Journal of Mathematical Physics {\bf 58}, 052204}~(2017).

\bibitem{Grinko2024}
Dmitry Grinko and Maris Ozols.
\newblock ``Linear programming with unitary-equivariant constraints''.
\newblock \href{https://dx.doi.org/10.1007/s00220-024-05108-1}{Communications in Mathematical Physics {\bf 405}, 278}~(2024).

\bibitem{Mozrzymas18}
Marek Mozrzymas, Michał Studziński, and Michał Horodecki.
\newblock ``A simplified formalism of the algebra of partially transposed permutation operators with applications''.
\newblock \href{https://dx.doi.org/10.1088/1751-8121/aaad15}{Journal of Physics A: Mathematical and Theoretical {\bf 51}, 125202}~(2018).

\bibitem{Nguyen23}
Quynh~T. Nguyen.
\newblock ``{The mixed Schur transform: efficient quantum circuit and applications}''~(2023).
\newblock  \href{http://arxiv.org/abs/2310.01613}{arXiv:2310.01613}.

\bibitem{Studzinski25}
Michał Studziński, Tomasz Młynik, Marek Mozrzymas, and Michał Horodecki.
\newblock ``{Irreducible matrix representations for the walled Brauer algebra}''~(2025).
\newblock  \href{http://arxiv.org/abs/2501.13067}{arXiv:2501.13067}.

\bibitem{Grinko23}
Dmitry Grinko, Adam Burchardt, and Maris Ozols.
\newblock ``Efficient quantum circuits for port-based teleportation''~(2024).
\newblock  \href{http://arxiv.org/abs/2312.03188}{arXiv:2312.03188}.

\bibitem{Nguyen24}
Quynh~T. Nguyen, Louis Schatzki, Paolo Braccia, Michael Ragone, Patrick~J. Coles, Fr\'ed\'eric Sauvage, Mart\'{\i}n Larocca, and M.~Cerezo.
\newblock ``Theory for equivariant quantum neural networks''.
\newblock \href{https://dx.doi.org/10.1103/PRXQuantum.5.020328}{PRX Quantum {\bf 5}, 020328}~(2024).

\bibitem{Grinko23gt}
Dmitry Grinko, Adam Burchardt, and Maris Ozols.
\newblock ``{Gelfand-Tsetlin basis for partially transposed permutations, with applications to quantum information}''~(2023).
\newblock  \href{http://arxiv.org/abs/2310.02252}{arXiv:2310.02252}.

\bibitem{Brundan08}
Jonathan Brundan and Alexander Kleshchev.
\newblock ``{Schur--Weyl} duality for higher levels''.
\newblock \href{https://dx.doi.org/10.1007/s00029-008-0059-7}{Selecta Mathematica {\bf 14}, 1--57}~(2008).

\bibitem{Marvian14}
Iman Marvian and Robert~W. Spekkens.
\newblock ``A generalization of {Schur--Weyl} duality with applications in quantum estimation''.
\newblock \href{https://dx.doi.org/10.1007/s00220-014-2059-0}{Communications in Mathematical Physics {\bf 331}, 431--475}~(2014).

\bibitem{Zhang15}
Lin Zhang.
\newblock ``Matrix integrals over unitary groups: An application of {Schur-Weyl} duality''~(2015).

\bibitem{Markiewicz25}
Marcin Markiewicz and Konrad Schlichtholz.
\newblock ``Unitary operator bases as universal averaging sets''~(2025).
\newblock  \href{http://arxiv.org/abs/2503.17091}{arXiv:2503.17091}.

\bibitem{Goodman09}
R.~Goodman and N.~R. Wallach.
\newblock ``Symmetry, representations, and invariants''.
\newblock Springer New York, NY. ~(2009).

\bibitem{HarrowPHD}
Aram~W. Harrow.
\newblock ``Applications of coherent classical communication and the schur transform to quantum information theory''~(2005).
\newblock  \href{http://arxiv.org/abs/quant-ph/0512255}{arXiv:quant-ph/0512255}.

\bibitem{SWC_Bacon06}
Dave Bacon, Isaac~L. Chuang, and Aram~W. Harrow.
\newblock ``{Efficient Quantum Circuits for Schur and Clebsch-Gordan Transforms}''.
\newblock \href{https://dx.doi.org/10.1103/PhysRevLett.97.170502}{Phys. Rev. Lett. {\bf 97}, 170502}~(2006).

\bibitem{SWC_Kirby18}
W.~M. Kirby and F.~W. Strauch.
\newblock ``{A practical quantum algorithm for the Schur transform}''.
\newblock \href{https://dx.doi.org/10.26421/QIC18.9-10-1}{Quantum Information and Computation {\bf 18}, 721–742}~(2018).

\bibitem{SWC_Krovi19}
Hari Krovi.
\newblock ``An efficient high dimensional quantum {S}chur transform''.
\newblock \href{https://dx.doi.org/10.22331/q-2019-02-14-122}{{Quantum} {\bf 3}, 122}~(2019).

\bibitem{Schlichtholz24}
Konrad Schlichtholz and Marcin Markiewicz.
\newblock ``Relativistically invariant encoding of quantum information revisited''.
\newblock \href{https://dx.doi.org/10.1088/1367-2630/ad2ffe}{New Journal of Physics {\bf 26}, 033018}~(2024).

\bibitem{Tung}
Wu-Ki Tung.
\newblock ``Group theory in physics''.
\newblock World Scientific Publishing Company. ~(1985).

\bibitem{Markiewicz21}
Marcin Markiewicz and Janusz Przewocki.
\newblock ``On construction of finite averaging sets for {SL(2,C)} via its {Cartan} decomposition''.
\newblock \href{https://dx.doi.org/10.1088/1751-8121/abfa44}{Journal of Physics A: Mathematical and Theoretical {\bf 54}, 235302}~(2021).

\bibitem{BZ17}
Ingemar Bengtsson and Karol Zyczkowski.
\newblock ``On discrete structures in finite {H}ilbert spaces''~(2017).
\newblock  \href{http://arxiv.org/abs/1701.07902}{arXiv:1701.07902}.

\bibitem{Dankert05}
Christoph Dankert.
\newblock ``Efficient simulation of random quantum states and operators''~(2005).
\newblock  \href{http://arxiv.org/abs/quant-ph/0512217}{arXiv:quant-ph/0512217}.

\bibitem{Dankert09}
Christoph Dankert, Richard Cleve, Joseph Emerson, and Etera Livine.
\newblock ``Exact and approximate unitary 2-designs and their application to fidelity estimation''.
\newblock \href{https://dx.doi.org/10.1103/PhysRevA.80.012304}{Phys. Rev. A {\bf 80}, 012304}~(2009).

\bibitem{Gross07}
D.~Gross, K.~Audenaert, and J.~Eisert.
\newblock ``Evenly distributed unitaries: On the structure of unitary designs''.
\newblock \href{https://dx.doi.org/10.1063/1.2716992}{Journal of Mathematical Physics {\bf 48}, 052104}~(2007).

\bibitem{Roy09}
A.~Roy and A.J. Scott.
\newblock ``Unitary designs and codes''.
\newblock Des. Codes Cryptogr. {\bf 53}, 13--31~(2009).

\bibitem{Webb16}
Z.~Webb.
\newblock ``{The Clifford group forms a unitary 3-design}''.
\newblock Q. Inf. Comp. {\bf 16}, 1379--1400~(2016).

\bibitem{Czartowski25}
Jakub Czartowski and Karol \ifmmode~\dot{Z}\else \.{Z}\fi{}yczkowski.
\newblock ``Quantum pushforward designs''.
\newblock \href{https://dx.doi.org/10.1103/PhysRevA.111.032433}{Phys. Rev. A {\bf 111}, 032433}~(2025).

\end{thebibliography}

\end{document}